\documentclass[review]{elsarticle}

\usepackage{lineno,hyperref}
\usepackage{amssymb}
\usepackage{amsmath}
\usepackage{amssymb}
\usepackage{amsxtra}
\usepackage{amsthm}
\usepackage[algoruled, linesnumbered]{algorithm2e}
\usepackage{graphicx}
\usepackage{color}

\newtheorem{prop}{Proposition}

\theoremstyle{definition}

\modulolinenumbers[5]

\begin{document}

\begin{frontmatter}

\title{A Graphical Bayesian Game for Secure Sensor Activation in Internet of Battlefield Things   \vspace{-1ex}}
\tnotetext[mytitlenote]{This research was sponsored by the Army Research Laboratory under Grant Number W911NF-17-1-0021. }

\author{Nof Abuzainab and Walid Saad}

\address{Wireless@VT, Department of Electrical and Computer Engineering, Virginia Tech, Blacksburg, VA, USA\\Emails:\{nof, walids\}@vt.edu}
%\email[mysecondaryaddress]{}

\begin{abstract}
In this paper, the problem of secure sensor activation  is studied for an Internet of Battlefield Things (IoBT) system in which an attacker compromises a set of the IoBT sensors for the purpose of eavesdropping and acquiring information about the battlefield. In the considered model, each IoBT sensor seeks to decide whether to transmit or not based on its utility. The utility of each sensor is expressed in terms of the redundancy of the data transmitted, the secrecy capacity and the energy consumed.  Due to the limited communication range of the IoBT sensors and their uncertainty about the location of the eavesdroppers, the problem is formulated as a graphical Bayesian game in which the IoBT sensors are the players.  Then, the utilities of the IoBT sensors are modified to take into account the effect of activating each sensor on the utilities of its neighbors, in order to achieve better system performance. The modified game is shown to be a Bayesian potential game, and a best response algorithm that is guaranteed to find a Nash equilibrium of the game is proposed. Simulation results show the tradeoff between the information transmitted by the IoBT sensors and the desired secrecy level. Simulation results also demonstrate the effectiveness of the proposed approach in reducing the energy consumed compared to the baseline in which all the IoBT sensors are activated. The reduction in energy consumption reaches up to $98\%$ compared to the baseline, when the number of sensors is $5000$.
\end{abstract}

%\begin{keyword}
%\texttt{elsarticle.cls}\sep \LaTeX\sep Elsevier \sep template
%\end{keyword}

\end{frontmatter}

\nolinenumbers

\section{Introduction}
Internet of Things (IoT) devices such as sensors, autonomous vehicles, and drones are projected to be integrated in military networks, forming the Internet of Battlefield Things (IoBT) \citep{ IoBTone, drone}.
In a military setting, the information collected from all IoBT devices can provide accurate and timely information about a certain battlefield environment, thus improving the efficiency of the military operations. However, given the massive number of IoBT devices, activating all the IoBT sensors simultaneously at a given time will incur significant cost, especially in terms of energy consumption.  Moreover, IoBT sensor measurements are often correlated based on their distance \cite{kraus2}, and, hence, tremendous amount of redundant information will be transmitted by the densely deployed IoBT sensors. Thus, there is a need to design efficient schemes for IoBT sensors activation, which determine the optimal set of activated sensors that minimizes both the energy consumed and the redundancy in the information transmitted.

Moreover, many IoBT devices have limited computational capabilities and therefore can not implement strong security measures. Thus, they are prone to multitude of security attacks and can be easily compromised by attackers, especially in a battlefield setting \cite{security}. One common attack scenario is one in which an adversary can compromise a set of the IoBT sensors in order to eavesdrop the information circulating within the IoBT\cite{IoBTone}. Therefore, each IoBT sensor must deliver its information securely from the eavesdroppers. In order to achieve security, physical layer security is favorable to be employed as it allows each sensor to deliver its information perfectly secure from eavesdroppers even in the presence of eavesdroppers with unlimited computational capabilities. Thus, each sensor will deliver its information with a rate which is equal to the secrecy capacity i.e. the maximum capacity such that the information is perfectly secured from the eavesdroppers. However, the energy overhead of activating a node with low secrecy capacity will be significant. Thus, the IoBT should take into account the secrecy capacity of the nodes in order to find the optimal activation policy. 

Several node activation schemes have been proposed for wireless sensor networks, the majority of which focus on energy efficiency \citep{actrout,actrout1,actrout3,actrout4,actener,actentr,actcom,actfish,actcov,actcov1}. In \cite{actrout,actrout1,actrout3,actrout4}, different routing protocols are proposed that control the sensors' sleep duration for energy efficient data collection in wireless sensor networks. In \cite{actener} and \cite{actentr}, distributed protocols for energy conservation in wireless sensor networks, using sensor activation schemes, are proposed. In \cite{actener},  a wireless sensor network with event driven traffic is considered, and a distributed protocol is proposed in which each node chooses to switch to the active mode only if the received signal power is above a certain threshold, thus significantly reducing the energy consumed in the network when no data is being transmitted. In \cite{actentr}, a distributed protocol to reduce energy consumption in a wireless sensor network is proposed which sets nodes transmitting redundant information to sleep mode, using information entropy and correlation graphs.
Centralized approaches for optimal node activation are proposed in \cite{actcom} and \cite{actfish}. The authors in \cite{actfish}  analyze the optimal set of activated sensors that  minimizes the Fischer information matrix of an unknown estimated parameter. In \cite{actcom}, the optimal set of activated nodes that minimizes energy and data redundancy is determined, using a compressive sensing scheme. Node activation schemes that maximize the area coverage while reducing data redundancy are developed in \cite{actcov} and \cite{actcov1}. In  \cite{actcov},  the proposed protocol reduces redundancy while maintaining coverage by allowing each sensor to join the network, if no other sensor is activated within the same communication range. 
In \cite{actcov1}, a predictive scheme is proposed that allows each sensor to decide whether or not to check redundancy at each time slot, thus reducing unnecessary redundancy checks and consequently the computational energy consumed.

On the other hand, there have been few schemes for secure sensor activation. In \cite{actmal}, a differential game is proposed for a wireless sensor network in which an attacker chooses the malware injection rate while the network
operator controls the sleep rate of the sensor nodes in addition to the recovery rate in order to
limit the spread of the malware. In \cite{actsec1}, an evolutionary game is formulated for a wireless sensor network in which each sensor chooses the transmission power that maximizes the secrecy rate while minimizing the energy consumed.
More recently, in \cite{actsec}, the authors investigated the security of a large scale wireless network containing both legitimate transmitters and eavesdroppers. The legitimate transmitters seek to find the optimal probability of node activation that maximizes their secrecy energy efficiency whereas the eavesdroppers find the probability of node activation that minimizes their energy efficiency. The problem is formulated as a noncooperative game between the set of transmitters and the set of eavesdroppers, and it is shown the magnitude of degradation of the secrecy energy efficiency of the legitimate transmitters due to the increase in the density of eavesdroppers.
%The limitations of the work in \cite{actsec} are that: 1) the proposed approach is not fully distributed, as the set of transmitters act as one player and the set of eavesdroppers act as another player 2) The proposed approach assumes full knowledge by the legitimate transmitters about the eavesdroppers, which is not realistic especially in scenarios where the eavesdroppers constitute compromised transmitters.
% in which the set of transmitters act as one player and the set of eavesdroppers act as another player

However, most of the existing approaches of sensor activation are either centralized such as in \cite{actcom} and \cite{actfish} or not fully distributed such as in  \cite{actsec}. Centralized approaches are not favorable for IoBT as they incur significant communication and computation overhead. 
Further, existing approaches either maximize the secrecy capacity such as in \cite{actsec1} and  \cite{actsec}  or minimize data redundancy such as in \cite{actentr,actcom, actfish}. Thus, there is no existing scheme that minimizes data redundancy while maximizing the secrecy capacity. Such a scheme is crucial for the successful operation of a mission critical system existing in adversarial environment, such as the IoBT, in which it is necessary to deliver the data as securely as possible. Besides, eliminating redundant transmitted data is essential for the IoBT in two main aspects. First, switching sensors sending redundant information to sleep mode helps in extending the lifetime of the IoBT devices. Second, removing redundant data significantly decreases congestion and therefore ensures the timely delivery of the information, which is a key requirement for IoBT.

The main contribution of this paper is a distributed scheme for sensor activation that is suitable for IoBT environment. The proposed scheme reduces the redundancy in the transmitted data while increasing the secrecy of the delivered information, using a graphical Bayesian game approach.
In particular, the key contributions are summarized as follows:
\begin{itemize}
\item We propose a distributed approach for IoBT sensor activation that aims at achieving energy efficiency, through deactivating sensors transmitting redundant information, while maintaining desirable secrecy. 
In terms of security, we consider a realistic attack in which an attacker compromises a set of the IoBT sensors for eavesdropping. Thus, in the proposed approach, each sensor node decides whether to transmit or not according to a utility function that captures the conditional entropy of the sensor's data given the measurements of its neighbors, its secrecy capacity, and the energy consumed. 
%\item In terms of security, 
\vspace{-0.2 cm}
\item We formulate the node activation problem as a graphical Bayesian game whose players are the IoBT sensors. The formulated game captures the limited communication range and resources of each IoBT sensor. In particular, the utility of each IoBT sensor depends on the actions of neighbors that are in its communication range.
The proposed Bayesian game also accounts for the uncertainty of each IoBT sensor about the location of the eavesdroppers by maintaining a belief about the probability of a neighbor being compromised. Therefore, the proposed approach does not require exact knowledge of the location of the eavesdroppers. 
\vspace{-0.2 cm}
\item In order to guarantee the existence of a pure strategy Nash equilibrium as well as to achieve better system performance,  we consider the alternate game where the repercussion utility for each IoBT sensor is used. The repercussion utililty allows each IoBT sensor to take into account the decrease in the utility in its neighbors upon its activation. We show that the modified game is a Bayesian potential game.
\vspace{-0.2 cm}
\item In the modified game, the utility of each sensor depends on the actions of the neighbor of its neighbors. Thus, in order to reach the Nash equilibrium, we propose a distributed learning scheme, using best response dynamics, suitable for the proposed game, in which each sensor broadcasts advertised action, from a sensor in its communication range, to its neighbors.
\vspace{-0.2 cm}
\item Simulation results show the tradeoff between the information transmitted by the IoBT sensors and the desired secrecy level. Simulation results also show the effectiveness of the proposed approach in reducing the energy consumed compared to the baseline in which all the IoBT sensors are activated. The reduction in energy consumption reaches up to $98\%$ compared to the baseline.
\end{itemize}
%distributed secrecy redundanc

\vspace{-0.3cm}
The paper is organized as follows. Section II describes the IoBT system model and the considered problem of secure sensor activation.
Section III presents the proposed IoBT graphical Bayesian game and its solution approach. Simulation results are shown in Section IV. Finally, conclusions are drawn in Section V.
%\vspace{-0.3 cm}
\section{System Model}
Consider an IoBT network that includes of a set $\mathcal{M}$ of $M$ sensors deployed to collect information about a certain phenomenon and subsequently transmit the collected information to their local sinks using an uplink wireless communication link. The IoBT network spans an $N\times N$ rectangular grid. Thus, the location of any node $i$ is represented by its coordinates 
$\textbf{l}_i=(x_i,y_i)$.
 Each sensor transmits information to its local sink using a transmit power value $P$ over orthogonal channels. Further, the simplified path loss model is used to model the attenuation of the signal transmitted by each sensor with distance. AWGN with variance $\gamma^2$ is present at each receiver. Thus, the received power by the local sink from sensor $i$'s transmission is given by
\begin{equation}
P_i=\frac{A|d_{ai}|^{-\alpha}P}{\gamma^2},
\end{equation}
where $A$ is the path loss coefficient, $\alpha$ is the path loss exponent, $d_{ai}$ is the distance between the local sink and sensor $i$ and is given by $d_{ai}=(x_{s_i}-x_a)^2+(y_{s_i}-y_a)^2$, where $\textbf{l}_a=(x_a,y_a)$ is the location of the local sink and  $\textbf{l}_{s,i}=(x_{s_i},y_{s_i})$ is the location of sensor $i$. %We de
\begin{figure}[t]{
	\centering
	\includegraphics[width=8 cm,height=4 cm,angle=0]{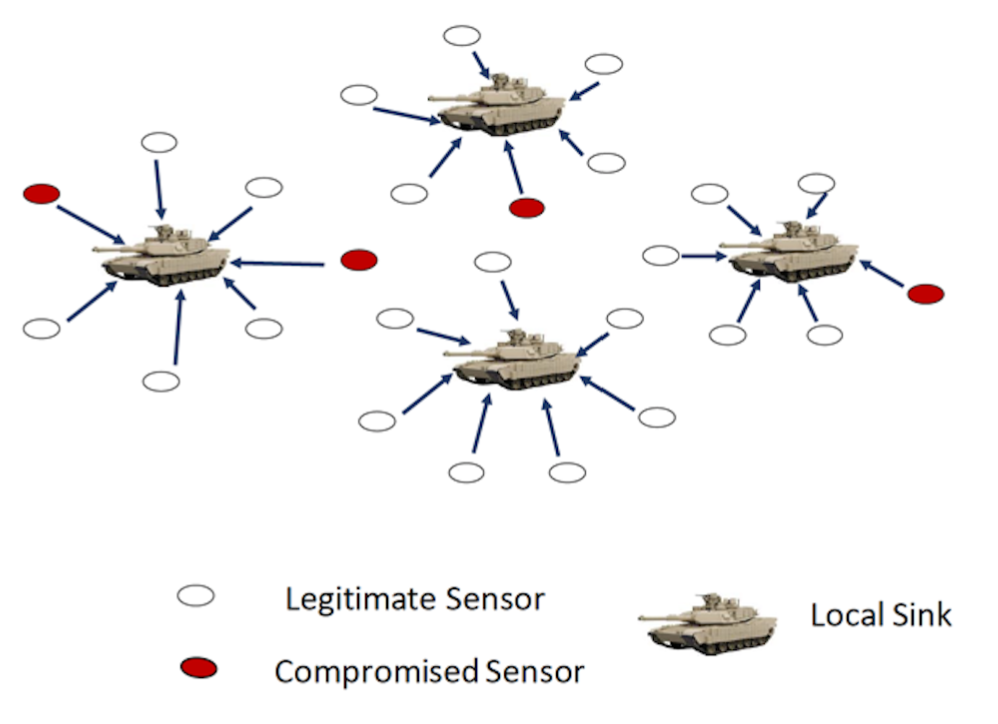}
	\caption{System Model}\label{sysmod}
	}\vspace{-0.1 cm}
\end{figure}
In a practical IoBT, the sensors' readings are correlated based on their locations \cite{kraus2}. In particular, the closer the sensors are located, the more likely they are collecting similar information.
%In order to capture such correlation,
 Thus, the information collected by the sensors is assumed to be distributed according to a multivariate Gaussian distribution which is often used to model spatial correlations between observations \cite{kraus}. Let $\textbf{Y}=(Y_i)_{i \in \mathcal{M}}$ be the vector of  sensors' observations, where $Y_i$ is a scalar whose value corresponds to sensor's $i$ measurement. The joint distribution of $\textbf{Y}$ is given by
\begin{equation}
\mathbb{P}(\boldsymbol{Y}=\boldsymbol{y})=\frac{1}{(2\pi)^{\frac{M}{2}}}e^{-\frac{1}{2}(\boldsymbol{y}-\boldsymbol{\mu})^T\Sigma^{-1}(\boldsymbol{y}-\boldsymbol{\mu})},
\end{equation}

where $\boldsymbol{\mu} $ is the mean and $\boldsymbol{\Sigma}$ is the covariance matrix  given by
 $\boldsymbol{\Sigma}=(\sigma_{ij})_{i,j \in \mathcal{M}}$
where each $\sigma_{ij}$ is a function of of the distance $d_{ij}$ between sensors $i$ and $j$. Typically, $\sigma_{ij}$ is chosen to decay exponentially with the distance using the radial basis function  \cite{kraus2} and is thus given by
\begin{equation}
\sigma_{ij}=\beta e^{-\frac{||(x_{s_i},y_{s_i})-(x_{s_j},y_{s_j})||^2}{2\lambda^2}},
\end{equation}
where the parameters $\beta$ and $\lambda$ are chosen to capture the correlation among the sensor measurements. The values of $\beta$ and $\lambda$ are usually estimated so that the Gaussian distribution matches the observed sensors' measurements.

Thus, in order to reduce redundancy in the transmitted information and to save energy, each IoBT sensor decides whether to activate and transmit information or switch to a sleep mode. Hence, the decision of IoBT sensor $i$ is made based on the conditional entropy $D(Y_i|\mathcal{W}_i)$ of its data $Y_i$ given the data generated by its neighbors $\mathcal{N}_i$ i.e., the IoBT sensors that are in its communication range, where $\mathcal{N}_i$ is the set of neighbors of sensor $i$ given by $\mathcal{N}_i=\{j \in \mathcal{M} | d_{ij}<r\}$, $r$ is the communication range, and $\mathcal{W}_i=\{Y_j, j \in \mathcal{N}_i\}$ is the set of data generated by  $\mathcal{N}_i$. The \emph{conditional entropy} is given by
\begin{equation}
D(Y_i|\mathcal{W}_i)=H(Y_i \cup \mathcal{W}_i)-H(\mathcal{W}_i), \label{cond}
\end{equation}
where, in general, the entropy of a multivariate Gaussian distribution of a set $\mathcal{K}$ of  $K$ random observations and of convariance matrix  $\boldsymbol{\Sigma}$ is given by
\begin{equation}
H(\mathcal{K} )=\frac{K}{2}+\frac{K}{2}\log(2\pi)+\frac{1}{2}\log|\boldsymbol{\Sigma}|.
\end{equation}
%and acquire knowledge about the IoBT network
In this IoBT, an attacker is interested in compromising a subset of the IoBT sensors for the purpose of eavesdropping and acquiring information about the IoBT. We denote by $\mathcal{X}$ the set of compromised sensors.

The uncompromised IoBT sensors do not have full knowledge of the set of compromised sensors. Thus, each IoBT sensor $i$ forms a belief $p_i$ about the compromised sensors in its neighborhood, where $p_{i}$ is the  probability that a sensor $j \in \mathcal{N}_i$ is compromised.
Hence, in order to ensure that the data is delivered to the base station as securely as possible, each sensor $i \in \mathcal{M}\setminus \mathcal{X}$ decides whether to transmit or switch to a sleep mode based on the secrecy capacity. It is assumed that the link between sensor $i$ and its local sink becomes insecure if any of the neighboring eavesdroppers successfully decode the message. Thus, the secrecy capacity of the link between sensor $i$ and its local sink is given by the channel capacity of the link between sensor $i$ and its local sink minus maximum of capacities that are individual achievable among the neighboring eavesdroppers \cite{secmult}. Hence, the secrecy capacity is derived as follows. Let $T_i$ be the type of sensor $i$ where $T_i=c$ if the sensor is compromised and $T_i=u$ otherwise.   Let $<\sigma_1, \sigma_2, ..., \sigma_{N_i}>$ be the ordered sequence of neighbors of sensor $i$ in $\mathcal{N}_i$ based on their distance to sensor $i$ and let $\boldsymbol{I}_i=(I_{ij})_{j \in \mathcal{N}_i}$ be an indicator vector with each entry $I_{ij}$ indicating whether sensor $i$ believes that its neighbor $j$ is compromised or not i.e. $I_{ij}=I_i(T_j=c)$. The secrecy capacity of the link between sensor $i$ and its local sink based on sensor $i$'s belief is 
%\begin{equation}
%\tilde{C}_{a,i}(\boldsymbol{p}_i)=C_{a,i}-\sum_{j =1}^{N_i} \prod_{k=1}^{j-1} (1-p_{i\sigma_k})p_{i\sigma_j}C_{i,\sigma_j} \label{secrecy}
%\end{equation}
%\begin{equation}
%\tilde{C}_{a,i}(\boldsymbol{I}_i)=C_{a,i}-\sum_{j =1}^{N_i} \prod_{k=1}^{j-1} (1-I_{i\sigma_k})I_{i\sigma_j}C_{i,\sigma_j} \label{secrecy}
%\end{equation}
\begin{equation}
\tilde{C}_{a,i}(T_i,\boldsymbol{T}_{-i})=\Big[C_{a,i}-\sum_{j =1}^{N_i} \prod_{k=1}^{j-1} (1-I_{i\sigma_k})I_{i\sigma_j}C_{i,\sigma_j}\Big]^+ \label{secrecy}
\end{equation}
where $z^+:=\max(z,0)$, $\boldsymbol{T}_{-i}$ is the vector of types of the neighbors of $i$, $C_{a,i}$ is the channel capacity between sensor $i$ and its local sink, and $C_{i,j}$ is the channel capacity between sensor $i$ and its neighbor $j$.
In our system, the channel capacity between any two nodes $r$ and $s$ is given by:
\begin{equation}
C_{rs}=W\log\Big(1+\frac{A|d_{rs}|^{-\alpha}P}{\gamma^2}\Big),
\end{equation}
where $W$ is the transmission bandwidth. Given this model, each uncompromised sensor $i$ will choose whether to transmit or switch to sleep mode in a way to maximize a utility function that we will define next. In essence, when sensor $i$ is activated, the utility can be defined as the product of the conditional entropy given by (\ref{cond}) and the achieved secrecy capacity in (\ref{secrecy}) minus the energy spent $E_i$ in each time instant.
When sensor $i$ switches to sleep mode, the utility is given by the energy saved $E_i$ minus the product of the conditional entropy and the secrecy capacity, which represents the loss of information due to not transmitting. Thus, the utility when transmitting corresponds to the cost when switching to the sleep node.  Let $a_i$ be a binary variable which is equal to one if sensor $i$ chooses to transmit and zero otherwise. Let $\boldsymbol{a}_{-i}$ be the vector of decisions made by the neighboring nodes of node $i$. Further, in order for the attack not to get detected by the IoBT, the compromised sensors choose their decision according to the same utility function as the uncompromised sensors. Then, the utility of  each sensor $i$ can be defined as follows:
\begin{equation}
  \hspace{-0.1 cm} U_i(a_i,\boldsymbol{a}_{-i},T_i,\boldsymbol{T}_{-i})=
\begin{cases}
D(Y_i|\mathcal{W}_i(\boldsymbol{a}_{-i}))\times\tilde{C}_{s,i}(T_i,\boldsymbol{T}_{-i})-E_i, & \hspace{-0.3 cm} \text{if }  a_i=1,\\
 E_i- D(Y_i|\mathcal{W}_i(\boldsymbol{a}_{-i}))\times\tilde{C}_{s,i}(T_i,\boldsymbol{T}_{-i}),              &\hspace{-0.3 cm} \text{if} \hspace{0.1 cm}  a_i=0,
\end{cases} \label{util}
\end{equation}
%where $\boldsymbol{T}_{-i}$ is the vector of types of the neighbors of sensor $i$.  
According to (\ref{util}), the utility of sensor $i$ depends on the decisions made by its neighbors in its communication range and not on the decisions made by all the IoBT sensors.
Further, the utility of sensor $i$ depends on its belief $p_i$ about the probability of a neighbor being compromised. Thus, the problem is formulated as a graphical Bayesian game \cite{graphical} with incomplete information \cite{gametheory}, as explained next.

\section{IoBT Activation Graphical Bayesian Game}
\subsection{Game Formulation}
The problem of activating the IoBT sensors is formulated as a graphical Bayesian game \cite{graphical,gametheory} i.e. a game on a graph defined by 
$\Big[<G,<T_i, \mathcal{A}_i, p_i, U_i>_{i \in \mathcal{M}}>\Big]$
where $G=(\mathcal{M},\mathcal{E})$ is an undirected graph representing the IoBT where each vertex $i \in \mathcal{M}$ is a player of the game and corresponds to an IoBT sensor, $\mathcal{E}$ is the set of edges where edge $e_{ij}$ exists between players $i$ and $j$ if they can directly communicate i.e. $d_{ij}<r$, $T_i$ is the type of player $i$, $A_i$ is the action set of player $i$,  $p_i(.|T_i)$ is the conditional probability distribution of player $i$ about the types of the other players in its neighborhood $\mathcal{N}_i$ given that player $i$'s type is $T_i$ and $U_i$ is the payoff of player $i$.

In our game, any player can be of two types $T_i \in \{c,u\}$,  where $T_i=c$ corresponds to a compromised node and $T_i=u$ corresponds to uncompromised node.  The action set of each node is $A_i=\{0,1\}$. The utility of each player $i$ $U_i(a_i,\boldsymbol{a}_{-i},T_i,\boldsymbol{T}_{-i})$ with respect to the actions and types of players in its neighborhood is given by (\ref{util}).

In order to solve our game, the graphical Bayesian Nash equilibrium (GBNE) is adopted as follows.
\subsection{Graphical Bayesian Nash Equilibrium}
%\begin{def}
A strategy profile $\boldsymbol{a}^*$ constitutes a GBNE if no player $i$  has the incentive to change its strategy given its neighbors' strategies i.e.
\begin{equation}
a^*_i \in \arg \max \sum_{\boldsymbol{T}_{-i}}p_i(\boldsymbol{T}_{-i})u_i(a_i,a^*_{-i}, T_i,\boldsymbol{T}_{-i}) \hspace{0.1 cm} \forall i \in \mathcal{M}
\end{equation}
\vspace{-0.8 cm}

In our originial game, a pure strategy Nash equilibrium is not guaranteed to exist. Thus, in order to ensure the existence of a pure strategy Nash equilibrium, we transform the game into a potential game by considering  the repercussion utility as done for allocation games in \cite{allo} and coalition formation games in \cite{coal}. The repercussion utility for each player in our IoBT activation graphical game  is defined as the utility of the player plus the change in the utility of the players in its neighborhood caused by its presence i.e.

\vspace{-0.5 cm}
\small
\begin{equation}
  \hspace{-0.1 cm} q_i(a_i,a'_{-i}, T_i, \boldsymbol{T}'_{-i})=
\begin{cases}
L_i(a_i,a'_{-i}, T_i, \boldsymbol{T}'_{-i}), & \text{if }  a_i=1,\\
 U_i(a_i,a_{-i}, T_i,\boldsymbol{T}_{-i}),              & \text{if} \hspace{0.1 cm}  a_i=0,
\end{cases} \label{repercussion}
\end{equation}
\normalsize
where $L_i(a_i,a'_{-i}, T_i, \boldsymbol{T}'_{-i})=U_i(a_i,a_{-i}, T_i,\boldsymbol{T}_{-i})
+\sum_{j \in \mathcal{N}_i}U_j(a_j,\boldsymbol{a}_{-j}, T_j, \boldsymbol{T}_{-j})-U_j(a_j,\boldsymbol{\tilde{a}}_{-j}, T_j,\boldsymbol{T}_{-j})$, $\boldsymbol{\tilde{a}}_{-j}$ corresponds to the vector of strategies of the neighbors of $j$ with $a_i$ replaced by $\tilde{a}_i=1-a_i$. Hence, each player, using the repercussion utility, acts cooperatively and chooses a decision that does not degrade significantly the utlities of the players in its neighborhood, thus improving the system performance. Such a cooperative behavior is appropriate in an IoBT since it relies only on information exchange among nodes within the same communication range, unlike conventional cooperative schemes that require communication among all the nodes, and, therefore, it does not incur significant communication overhead. Besides, the repercussion utility defined in (\ref{repercussion}) assumes that a given player $i$ has information about the utilities of the neighbors of its neighbors. Thus, initially, IoBT sensors in the same communication range exchange information about their neighbors prior to computing their optimal strategies.
In the modified graphical game, the new set of players influencing player $i$'s utility include the neighbors of its neighbors i.e.
$\mathcal{N}'_i=\mathcal{N}_i  \cup_{j \in \mathcal{N}_i}\mathcal{N}_j$ and the vectors $\boldsymbol{a}'_{-i}$ and $\boldsymbol{T}'_{-i}$ are the strategy vector and type vector of players in $\mathcal{N}'_i$, respectively.

The following proposition shows that the our modified game with  the utility defined in (\ref{repercussion}), based on repercussion utilities, is a Bayesian potential game.

\begin{prop}
Our IoBT graphical game with utilities defined in (\ref{repercussion}) is a Bayesian potential game where the potential function V is the sum of the player original utililties  i.e.
$V(a_i,\boldsymbol{a}_{-i},\boldsymbol{T})=\sum_{i \in \mathcal{M} }U_i(a_i,\boldsymbol{a}_{-i},T_i,\boldsymbol{T}_{-i}),$
where $\boldsymbol{T}$ is the vector of types of all players.
\label{prop1}
\end{prop}

\begin{proof}
Fix any player $i$ and let $\boldsymbol{T}$ is the vector of types of all players. For any $\boldsymbol{T}$ and $\boldsymbol{a}_{-i}$,  we have
\begin{eqnarray}
&&V(a_i,\boldsymbol{a}_{-i},\boldsymbol{T})-V(\tilde{a}_i,\boldsymbol{a}_{-i},\boldsymbol{T})=U_i(a_i,\boldsymbol{a}_{-i},T_i,\boldsymbol{T}_{-i})-U_i(\tilde{a}_i,\boldsymbol{a}_{-i},T_i,\boldsymbol{T}_{-i})\nonumber\\
&&\hspace{3 cm}+\sum_{j \in \mathcal{M},j\neq i}(U_j(a_j,\boldsymbol{a}_{-j},T_j,\boldsymbol{T}_{-j})-U_j(a_j,\boldsymbol{\tilde{a}}_{-j},T_j,\boldsymbol{T}_{-j}))\nonumber\\
&&=U_i(a_i,\boldsymbol{a}_{-i},T_i,\boldsymbol{T}_{-i})-U_i(\tilde{a}_i,\boldsymbol{a}_{-i},T_i,\boldsymbol{T}_{-i})\nonumber\\
&&\hspace{0.1 cm}+\sum_{j \in \mathcal{N}_i}(U_j(a_j,\boldsymbol{a}_{-j},T_j,\boldsymbol{T}_{-j})-U_j(a_j,\boldsymbol{\tilde{a}}_{-j},T_j,\boldsymbol{T}_{-j}))\nonumber\\
&&=q_i(a_i,\boldsymbol{a}'_{-i},T_i,\boldsymbol{T}'_{-i})-q_i(\tilde{a}_i,\boldsymbol{a}'_{-i},T_i,\boldsymbol{T}'_{-i}), \label{pot}
\vspace{-0.4 cm}
\end{eqnarray}
\normalsize
where the third equality follows from the fact that player $i$'s action affects only its neighbors and the last equality is obtained based on the definition of $q_i(a_i,\boldsymbol{a}_{-i},T_i,\boldsymbol{T}'_{-i})$ in (\ref{repercussion}).

Thus, according to (\ref{pot}),  our IoBT graphical game with utilities defined in (\ref{repercussion}) is a  Bayesian potential game with the potential function equal to the sum of original utilities of the players. Therefore, it admits a pure strategy Nash equilibrium \cite{pot}. 
\end{proof}
In order to find a pure strategy Nash equilibrium, we present a learning algorithm, defined by Algorithm 1, based on the best response dynamics and tailored to the characteristics of our IoBT graphical Bayesian game.
The proposed learning scheme allows the IoBT sensors to find their Nash equilibrium strategy in a distributed fashion while taking into account the limited communication range of the IoBT sensors. In particular, in each iteration of the algorithm, each IoBT sensor $i$ computes its optimal strategy given the current strategy profile $\boldsymbol{a}'_{-i}$.
Then, sensor $i$ broadcasts its optimal strategy to its neighbors which consequently broadcast sensor $i$ strategy to their neighbors. The process is repeated until convergence.

The proposed learning scheme is guaranteed to converge to a pure strategy Nash equilibrium since the modified graphical Bayesian game is a Bayesian potential game according to Proposition \ref{prop1}.

The performance of our IoBT graphical game, at equilibrium, is assessed next.
\begin{algorithm}[t]
%	\small The BS chooses an initial feasible vector of time allocations $\boldsymbol{\tau}_0$ and broadcasts it to its associated devices
	
	\scriptsize Each sensor $i \in \mathcal{M}$ initializes its strategy vector $\boldsymbol{a}'_{-i}$.
	
	\Repeat{convergence}{
		
		\ForEach{sensor $i \in 1, ...,M$}{
			Sensor $i$ computes its best response. %according to proposition \ref{propGNEHTD} if $i$ is HTD and according to proposition \ref{propGNEMTD}  if $i$ is MTD.
			
			Sensor $i$ broadcasts its newly computed best response to its neighbors in $\mathcal{N}_i$.

                       Each sensor $j$ in $\mathcal{N}_i$ broadcasts the best response of sensor $i$ to its neighbors in $\mathcal{N}_j$
			
			All devices in $\mathcal{N}'_i$ update their strategy vectors. }
		
		%Set $k \rightarrow k+1$
		
		%			}
		%		}
	}
	\caption{Nash Equilibrium Learning Algorithm for the IoBT graphical Bayesian Game}\label{alg1}
\end{algorithm}

\section{Simulation Results}
\vspace{-0.3 cm}
For our results, we consider a rectangular area of size $100$ m $\times$ $100$ m in which the sensors are randomly deployed according to a uniform distribution. The values of the parameters considered are:
$W=20$ MHz, $\sigma^2=-90.2$ dBm, $\beta=1$, $\lambda=1$, $P=0.1$ W, $r=2$ m, and $T=1$ ms. The probability that a given sensor is compromised is assumed to be the same for all sensors i.e. $p_i=p_e \hspace{0.1 cm} \forall i \in \mathcal{M}$ where $p_e$ is a constant. Two scenarios are considered:
\begin{itemize}
\vspace{-0.2 cm}
\item The probability that a sensor is compromised is varied between $0$ to $1$ in steps of $0.1$. Two values of the total number of sensors are considered: $M=1000$ and $M=3000$.
\vspace{-0.2 cm}
\item The total number of sensors is varied between $500$ to $5000$ in steps of $500$. Two values of the probability that a sensor is compromised are considered: $p_e=0.1,0.5$.
\end{itemize}
\vspace{-0.2 cm}
 For comparison, a baseline in which all the sensors are activated is considered, and for the considered scenarios, 1000 independent simulation runs are performed. Then,  at equilibrium, the average number of activated sensors, the average entropy of the activated sensors, and the average percentage decrease in the energy consumed using the equilibrium solution compared to the baseline are computed. Algorithm 1 converges in at most $5$ iterations.
\begin{figure}
\centering
%\begin{minipage}[t]{\dimexpr.5\textwidth-1em}
\begin{minipage}[t]{ 0.45\textwidth}
 \hspace{-0.4 cm} \includegraphics[width=1.1\linewidth]{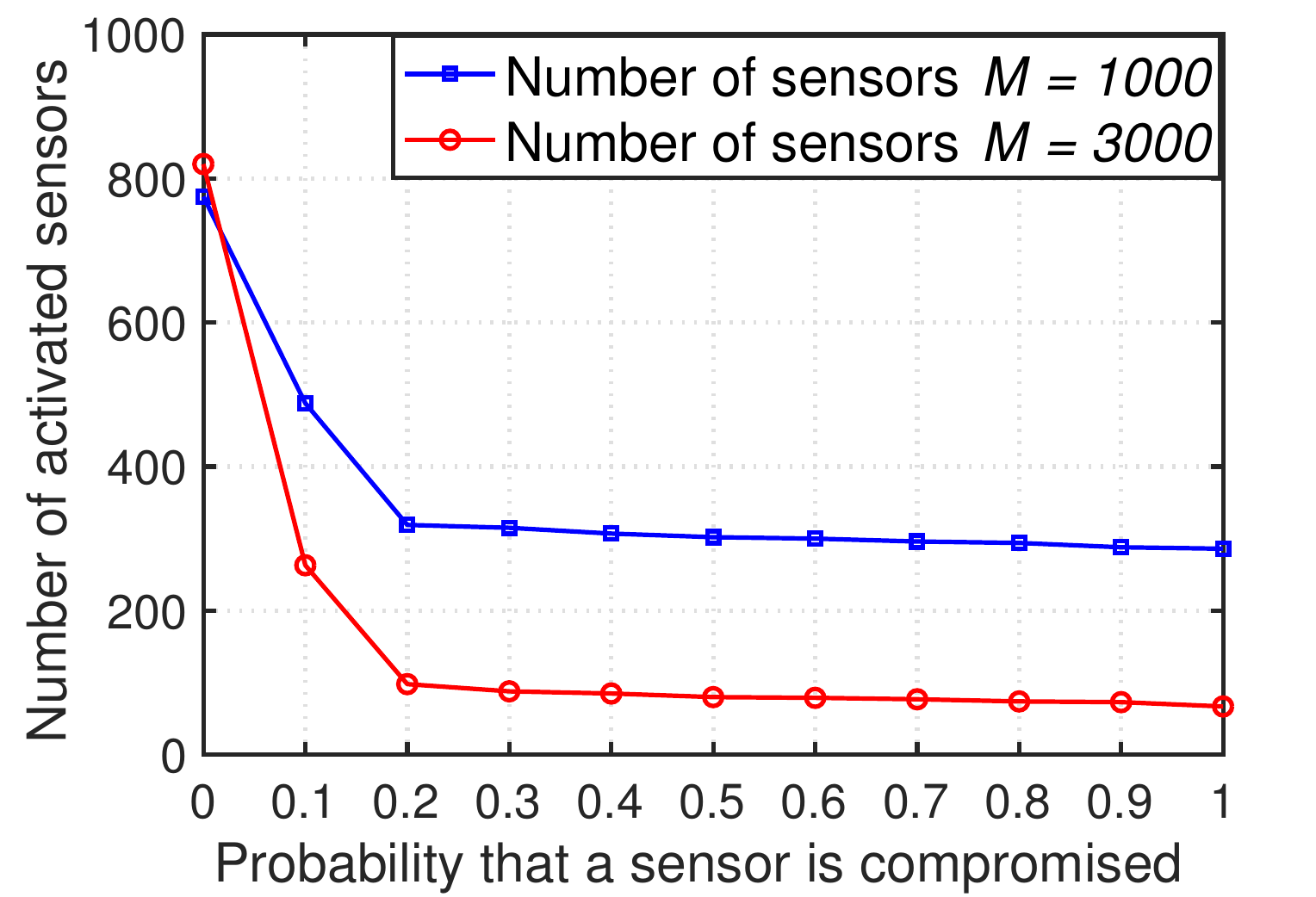}
  \caption{Number of activated sensors versus the probability that a sensor is compromised for  different values of the number of sensors.}
  \label{actprob}
\end{minipage}%
\qquad 
\begin{minipage}[t]{0.45 \textwidth}
  \centering
  \includegraphics[width=1.1\linewidth]{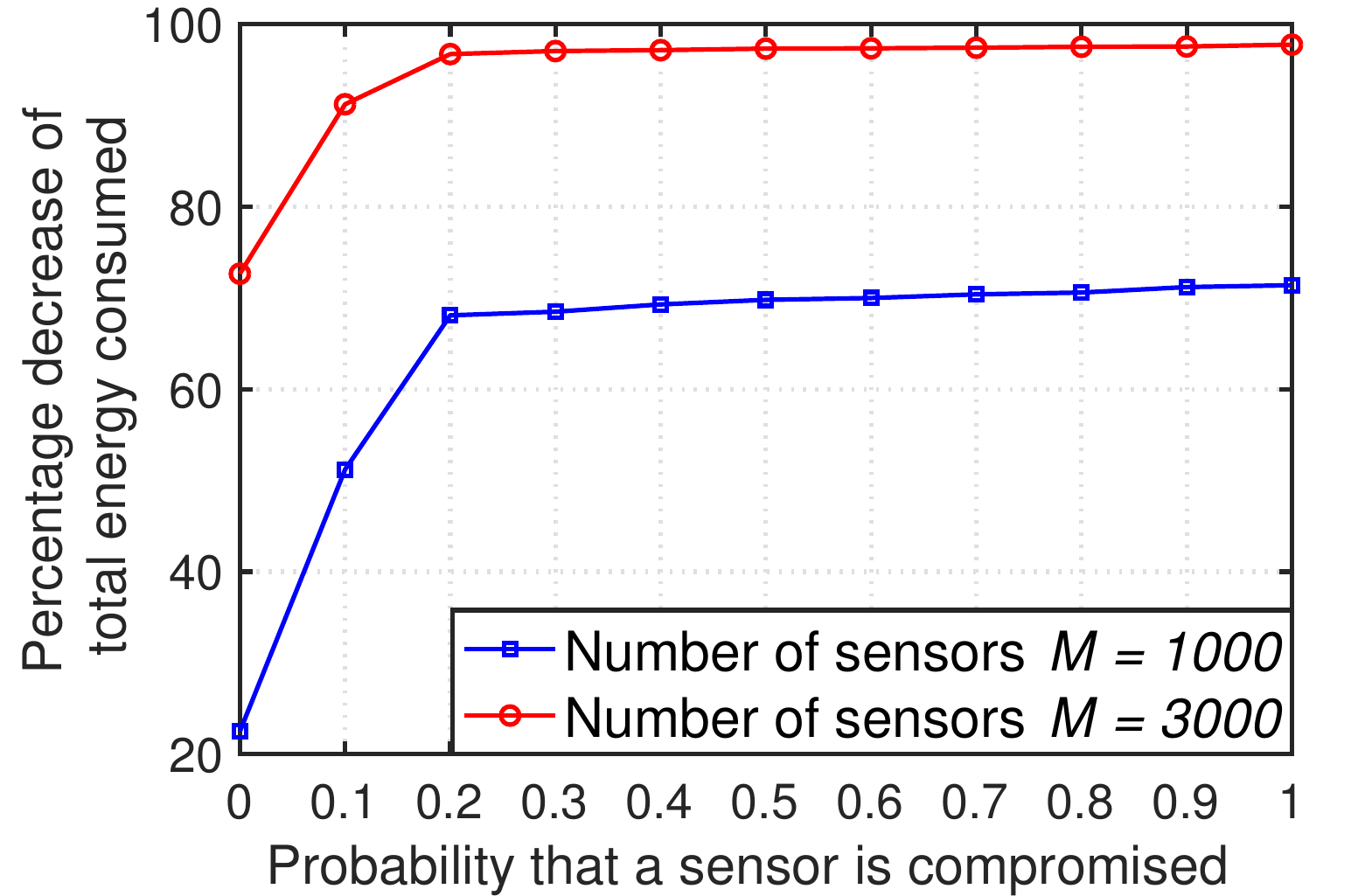}
  \caption{Percentage decrease of energy consumed versus the probability that a sensor is compromised for different numbers of sensors.}
  \label{perdecrprob}
\end{minipage}
\end{figure}
%\begin{figure}[t]{
%	\centering
%	\includegraphics[width=9 cm,height=6 cm,angle=0]{actprob.pdf}
%	\caption{Number of activated sensors versus the probability that a sensor is compromised for  different values of 
%the number of sensors.}\label{actprob}
%	}
%\end{figure}
%
%\begin{figure}[t]{
%	\centering
%	\includegraphics[width=9 cm,height=6 cm,angle=0]{perdecrprob.pdf}
%	\caption{Percentage decrease of energy consumed versus the probability that a sensor is compromised for  different values of 
%the number of sensors.}\label{perdecrprob}
%	}
%\end{figure}

\begin{figure}[t]{
	\centering
	\includegraphics[width=6.5 cm,height=4 cm,angle=0]{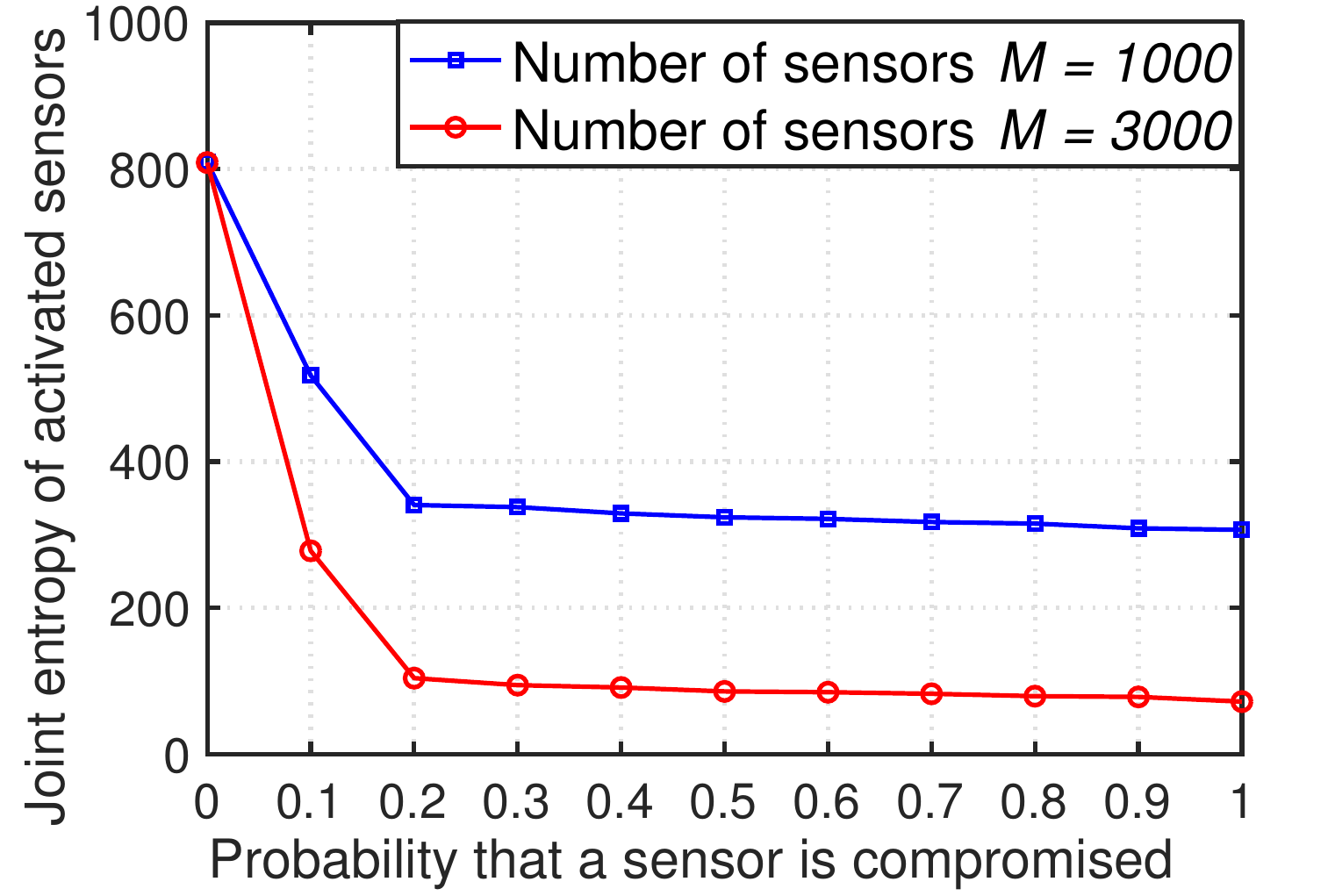}
	\caption{Joint entropy of activated sensors versus the probability that a sensor is compromised for  different numbers of sensors.}\label{entract}
	}\vspace{-0.5 cm}
\end{figure}
Fig. \ref{actprob} shows the average number of activated sensors resulting from the equilibrium solution versus the probability that a sensor is compromised $p_e$ when the total number of sensors $M$ is $1000$ and $3000$, respectively. For both considered values of $M$, the number of activated sensors decreases as the probability that a sensor is compromised increases. This is because the increase in the probability of an eavedropper decreases the secrecy capacity which consequently results in decreasing the utility of transmitting. 
Also, Fig. \ref{actprob} shows that when $M$ increases from $1000$ to $3000$, the number of activated sensors decreases for all considered values of $p_e$. This is due to the fact that, as the number of sensors increases, the proportion of compromised sensors increases which causes the secrecy capacity to decrease, and hence, fewer sensors are activated.
%activated=positive secrecy capacity

Fig. \ref{perdecrprob} shows the percentage decrease of energy consumed using the equilibrium solution compared to the baseline versus the probability that a sensor is compromised when the total number of sensors $M$ is $1000$ and $3000$, respectively. As shown in Fig. \ref{perdecrprob}, the percentage decrease in the energy consumed increases significantly as the probability that a sensor is compromised increases. This is due to the significant decrease in the number of activated sensors as shown in Fig. \ref{actprob}. 
The percentage decrease in the energy consumed reaches up to $97\%$ when $M=3000$ and $p_e=1$.

Fig. \ref{entract} shows the average joint entropy of the activated sensors using the equilibrium solution versus the probability that a sensor is compromised when the total number of sensors $M$ is $1000$ and $3000$, respectively. Fo all considered values of $p_e$, the joint entropy is positive, yet  it decreases with the probability that a sensor is compromised, due to the decrease in the number of activated sensors as shown in Fig. \ref{actprob}.
As for the baseline, the joint entropy is $- \infty$. This is because the considered number of sensors corresponds to a dense deployment of the IoBT and all the sensors are activated in the baseline. Thus, the sensor's data given the remaining sensor's measurements becomes deterministic. Thus, Fig. \ref{entract} demonstrates that the proposed approach is effective in reducing the redundancy in the transmitted information. Fig. \ref{entract} also shows the tradeoff between the secrecy level and the information transmitted by the IoBT sensors.
\begin{figure}
\centering
%\begin{minipage}[t]{\dimexpr.5\textwidth-1em}
\begin{minipage}[t]{ 0.45\textwidth}
 \hspace{-0.4 cm} \includegraphics[width=1.1\linewidth]{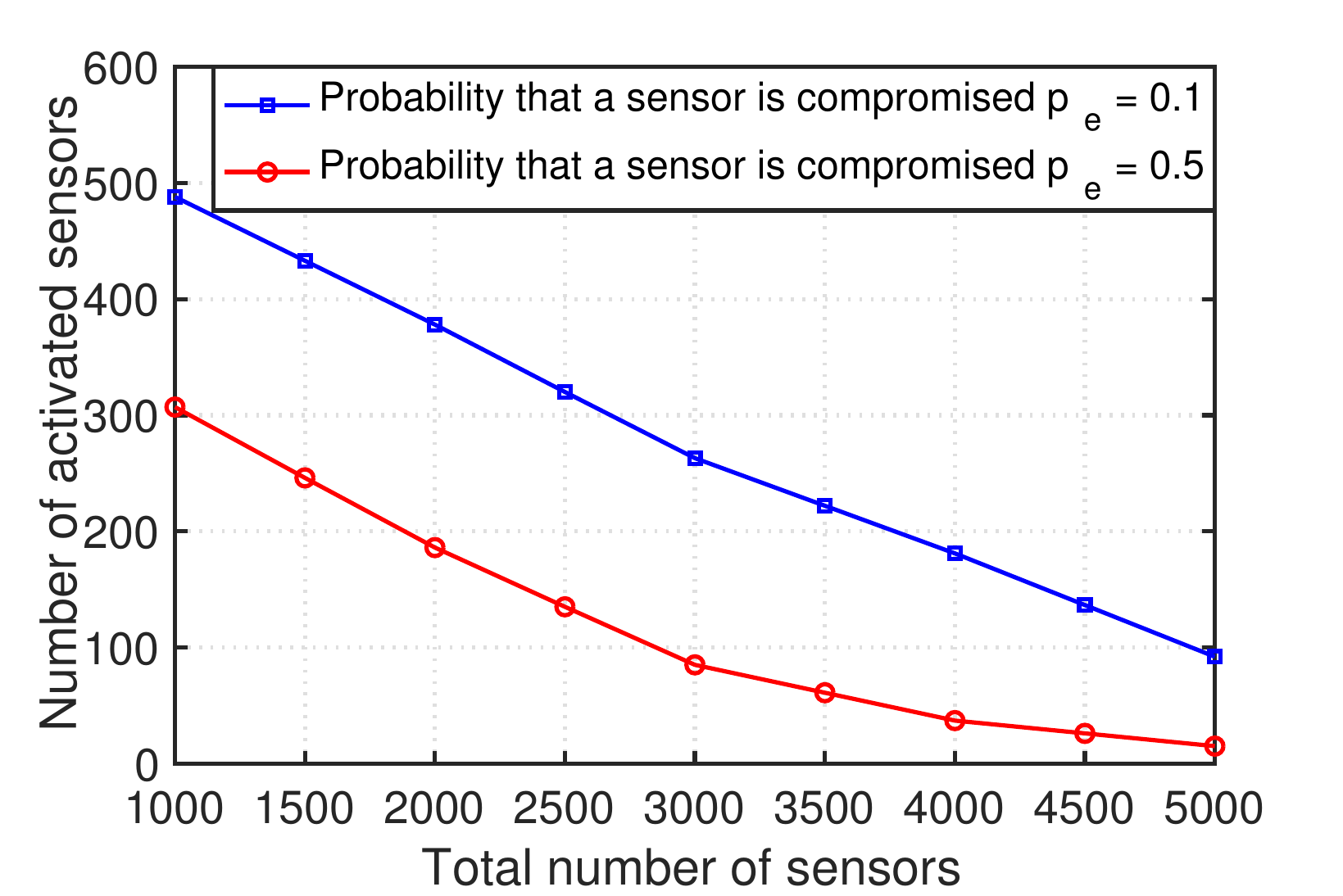}
  \caption{Number of activated sensors versus the total number of sensors for different values of the probability that a sensor is compromised.}
  \label{actnum}
\end{minipage}%
\qquad 
\begin{minipage}[t]{0.45 \textwidth}
  \centering
  \includegraphics[width=1.1\linewidth]{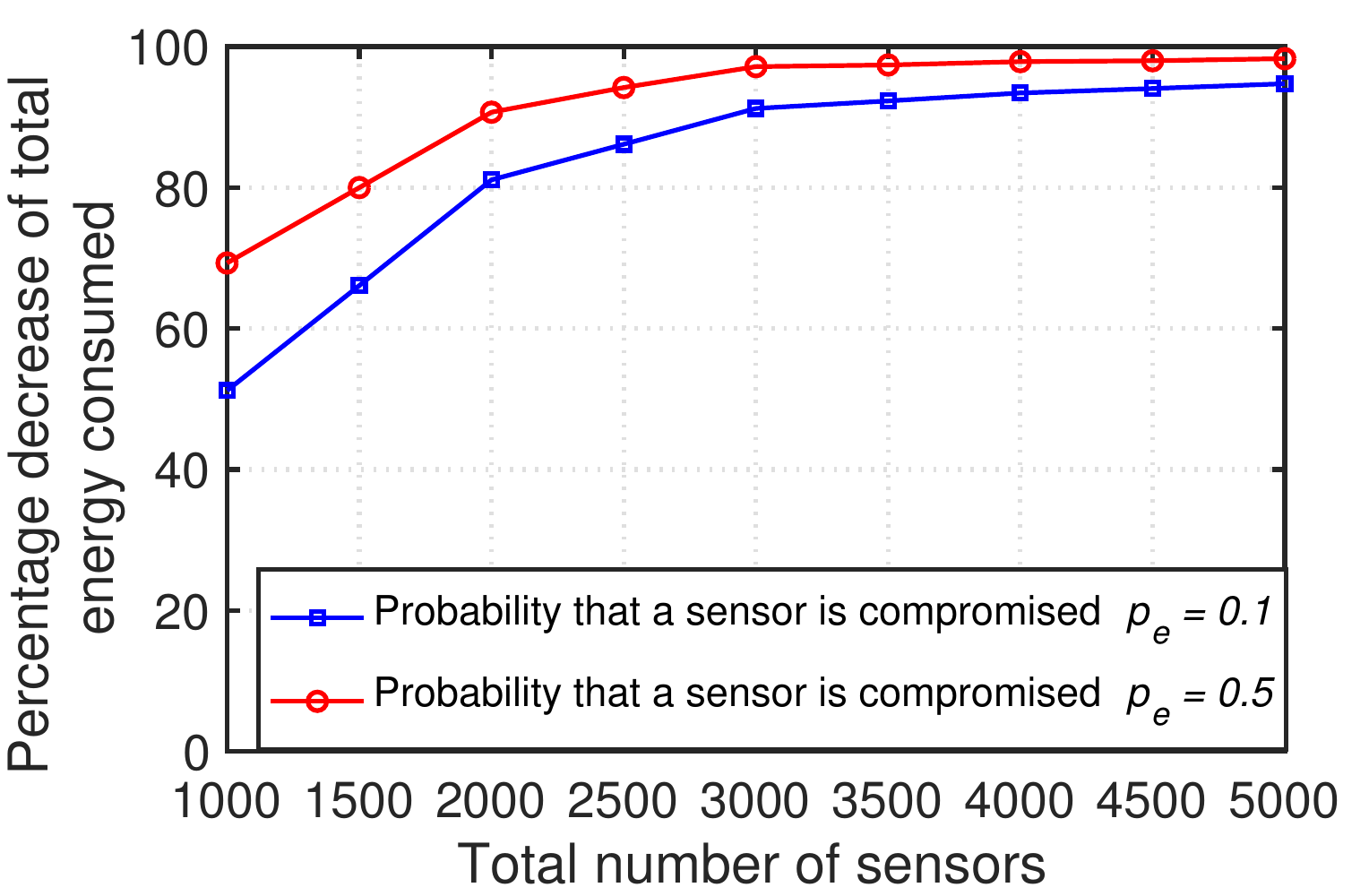}
  \caption{Percentage decrease of energy consumed versus the total number of sensors for different values of the probability that a sensor is compromised.}
  \label{perdecrnum}
\end{minipage}
\end{figure}
\begin{figure}
\centering
\begin{minipage}[t]{ 0.45\textwidth}
 \hspace{-0.4 cm} \includegraphics[width=1.1\linewidth]{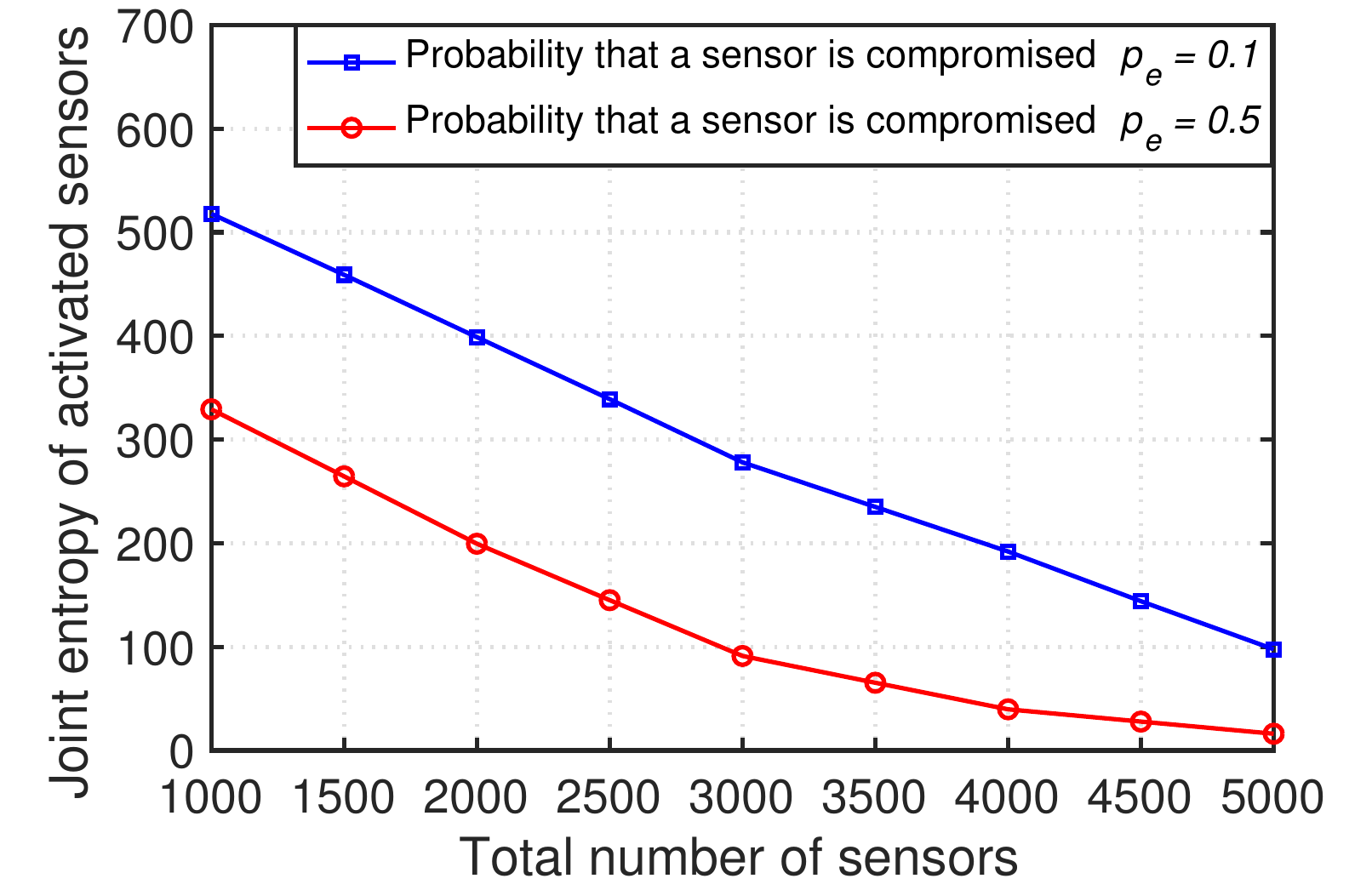}
  \caption{Joint entropy of activated sensors versus total number of sensors for different values of the probability that a sensor is compromised}
  \label{entrnum}
\end{minipage}%
\qquad 
\begin{minipage}[t]{0.45 \textwidth}
  \centering
  \includegraphics[width=1.05\linewidth]{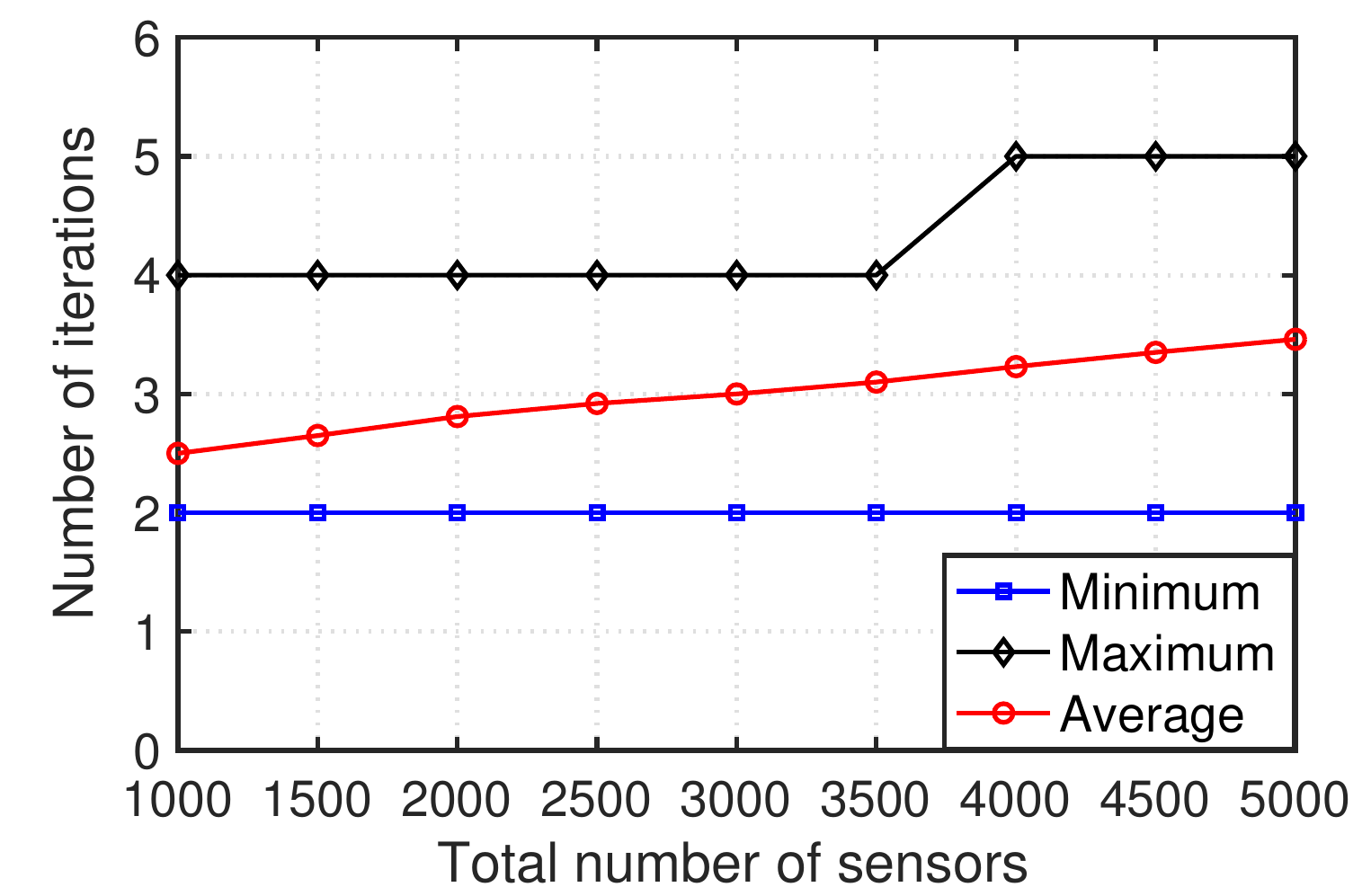}
  \caption{Number of iterations until convergence versus total number of sensors.}
  \label{iobtiter}
\end{minipage}
\end{figure}

%\begin{figure}
%\centering
%\begin{minipage}{.5\textwidth}
%  \centering
%  \includegraphics[width=1\linewidth]{entrnum.pdf}
%  \caption{figure}
%  \label{entrnum}
%\end{minipage}%
%\begin{minipage}{.5\textwidth}
%  \centering
%  \includegraphics[width=1\linewidth]{iobtiter.pdf}
%  \caption{figure}
%  \label{iobtiter}
%\end{minipage}
%\end{figure}

%\begin{figure}[t]{
%	\centering
%	\includegraphics[width=6.5 cm,height=4 cm,angle=0]{entrnum.pdf}
%	\caption{Joint entropy of activated sensors versus total number of sensors for different values of the probability that a sensor is compromised
%.}\label{entrnum}
%	}
%\end{figure}
%
%\begin{figure}[t]{
%	\centering
%	\includegraphics[width=6.5 cm,height=4 cm,angle=0]{iobtiter.pdf}
%	\caption{Number of iterations until convergence versus total number of sensors 
%.}\label{iobtiter}
%	}
%\end{figure}
\vspace{-0.5 cm}
Fig. \ref{actnum} shows the average number of activated sensors at the GBNE as a function of the total number of sensors when $p_e$ is $0.1$ and $0.5$, respectively. Fig. \ref{actprob} shows that, for both considered values of the probability $p_e$, the number of activated sensors decreases as the total number of devices increases. This is due to the fact that, as the density of sensors increases, the secrecy capacity decreases, which consequently increases the utility of each sensor and causes fewer number of sensors to be activated.

Fig. \ref{perdecrnum} shows the percentage decrease in the  total energy consumed using the equilibrium solution compared to the baseline versus the total number of sensors $M$ when $p_e$ is $0.1$ and $0.5$, respectively. Fig. \ref{perdecrnum} shows that the percentage decrease in the energy consumed increases with $M$ due to the decrease in the number of activated sensors, as shown in Fig.  \ref{actnum}. Fig. \ref{perdecrnum} confirms the significant decrease in the total energy consumed using the equilibrium solution compared to the baseline. The decrease in the total energy consumed reaches up to $98\%$ when $M=5000$ and $p_e=1$.

Fig.  \ref{entrnum} shows the  average joint entropy of the activated sensors at the GBNE as a function of the total number of sensors $M$ when the probability that a sensor is compromised $p_e$ is $0.1$ and $0.5$, respectively. As shown in Fig.  \ref{entrnum} for both considered values of $p_e$, the joint entropy decreases as the number of devices increases, due to the decrease in the number of activated sensors according to Fig. \ref{actnum}. As for the baseline, the value of joint entropy remains $- \infty$ as the density of the devices increases. Thus, Fig. \ref{entrnum} shows that, using the equilibrium solution, the redundancy of the transmitted information is reduced significantly compared to the baseline, as $M$ increases.

Fig. \ref{iobtiter} shows the maximum, minimum, and average number of iterations spent until Algorithm 1's convergence as a function of the total number of sensors.  The value of the probability that a sensor is compromised $p_e$ is $0.1$. According to Fig. \ref{iobtiter}, the maximum number of iterations is $4$ when the total number of  sensors is less than or equal to $3500$. When the number of sensors becomes greater than or equal to $4000$, the maximum number of iterations increases by only one iterations and becomes $5$.
As for the average number of sensors, its value is $2.5$ when the number of sensors is $1000$. Then, the average number of iterations increases slightly with the number of sensors until it reaches $3.55$ when the number of sensors is $5000$.
Fig. \ref{iobtiter} also shows that the minimum number of iterations is $2$ for all considered number of sensors. Thus, Fig. \ref{iobtiter} demonstrates the fast convergence of Algorithm 1 for large number of devices.
\vspace{-0.5 cm}
\section{Conclusion}
\vspace{-0.3 cm}
In this paper, we have considered the problem of secure sensor activation in the Internet of Battlefield Things in which each sensor decides whether to transmit or not in order to maximize its payoff which is a function of the secrecy rate as well as the redundancy of the transmitted information.
We have formulated the problem as a graphical Bayesian game and have shown that the proposed game is Bayesian potential game.  We have proposed a learning algorithm that is suitable for our IoBT graphical game and that is guaranteed to converge to a Nash equilibrium.
Our results have shown the tradeoff between the information transmitted by the IoBT sensors and the desired secrecy level. Our results have also shown the effectiveness of the proposed approach in reducing the energy consumed compared to the baseline in which all the IoBT sensors are activated. The reduction in energy consumption reaches up to $98\%$ compared to the baseline, when the number of sensors is $5000$. For future work, we will extend this work to include other attack types, such as those in which incorrect information is transmitted.
\vspace{-0.2 cm}
\section*{References}
%\baselinestretch{0.95}
%\bibliography{mybibfile}
\smaller{
}
\end{document}